\documentclass[sigconf]{acmart-primary/acmart}

\usepackage[inline]{enumitem}
\usepackage{caption}
\usepackage{subcaption}
\usepackage{balance}
\usepackage[binary-units]{siunitx}
\usepackage{graphicx}
\usepackage{hyperref}
\usepackage{cleveref}
\usepackage{makecell}
\usepackage{listings}
\usepackage{xspace}
\usepackage{natbib}
\usepackage{xspace}
\usepackage[title]{appendix}

\usepackage{hyperref}
\usepackage{cleveref}

\newtheorem{theorem}{Theorem}

\newcommand{\aigc}[1]{\textsc{AIGC}}
\newcommand{\AIGC}[1]{\textsc{AIGC}}
\newcommand{\sysname}[1]{OnePiece}
\newcommand{\NM}[1]{\textsc{NM}}

\usepackage[noend]{algorithm, algorithmic}
\newenvironment{myprotocol}{
    \hrule
    \smallskip
    \footnotesize
    \algsetup{linenosize=\tiny}
    \begin{algorithmic}[1]

        \makeatletter
        \newcommand{\EVENT}[1]{\STATE \textbf{event} ##1 \textbf{do}\begin{ALC@g}}
                \newcommand{\ENDEVENT}{\end{ALC@g}}
        \makeatother

        \makeatletter
        \newcommand{\FUNC}[2]{\STATE \textbf{function} \textbf{##1} (##2) \begin{ALC@g}}
                \newcommand{\ENDFUNC}{\end{ALC@g}}
        \makeatother

        }{
    \end{algorithmic}
    \smallskip
    \hrule
}

\AtBeginDocument{}

\begin{document}

\title{\sysname{}: A Large-Scale Distributed Inference System with RDMA for Complex AI-Generated Content (AIGC) Workflows}

\author{June Chen, Neal Xu, Gragas Huang, Bok Zhou, Stephen Liu}
\affiliation{\institution{Wechat Tencent} \city{} \country{}}




\begin{abstract}
   The rapid growth of AI-generated content (AIGC) has enabled high-quality creative production across diverse domains, 
yet existing systems face critical inefficiencies in throughput, resource utilization, and scalability under concurrent workloads. 
This paper introduces \sysname{}, a large-scale distributed inference system with RDMA optimized for multi-stage AIGC workflows. 
By decomposing pipelines into fine-grained microservices and leveraging one-sided RDMA communication, 
\sysname{} significantly reduces inter-node latency and CPU overhead while improving GPU utilization. 
The system incorporates a novel double-ring buffer design to resolve deadlocks in RDMA-aware memory access without CPU involvement. 
Additionally, a dynamic Node Manager allocates resources elastically across workflow stages in response to real-time load. 
Experimental results demonstrate that \sysname{} reduces GPU resource consumption by $16\times$ in Wan2.1 image-to-video generation 
compared to monolithic inference pipelines, offering a scalable, fault-tolerant, and efficient solution for production AIGC environments.
\end{abstract}

\maketitle

\vspace{-2mm}
\section{Introduction}

The advent of AI-generated content (\AIGC{}) has dramatically lowered barriers 
to creative expression and content production, 
empowering individuals and organizations to produce high-quality materials 
with unprecedented speed and scale. 
\AIGC{} is being applied across diverse fields such as automated journalism, 
personalized advertising, video game design, 
and virtual reality content creation. 
Central to this progress are diffusion models, 
whose remarkable generative capabilities are increasingly surpassing those of GANs and auto-regressive 
Transformers. 
These models excel not only in image generation but are also driving innovation 
in video-related research and production. 
Video diffusion models
support a wide range of applications—including video synthesis, 
editing, and enhanced video understanding. 
By providing non-experts with advanced generative tools, 
these technologies are helping democratize content creation. 
As a result, \AIGC{} is fundamentally transforming traditional workflows within creative industries.

\AIGC{} workflows typically integrate a variety of heterogeneous models,
such as Variational Autoencoders (VAEs)~\cite{dai2019diagnosing,doersch2016tutorial}, 
diffusion models~\cite{croitoru2023diffusion,wijmans1995solution,yang2023diffusion}, 
and text encoders~\cite{ni2021sentence},each fulfilling a distinct role 
within the content generation pipeline. 
In practice, specialized models like LoRA~\cite{hu2022lora,8474715} are often incorporated to further enhance output quality, 
leading to increasingly complex and multi-stage processing architectures. 
This multi-model design, however, often introduces significant latency 
and high GPU resource demands. 
For example, the WAN2.1 model~\cite{wan2.1} requires around 32GB of GPU memory distributed across 8 GPUs.

Despite significant advances in reducing per-task latency, 
a critical challenge persists in real-world \AIGC{} deployments: 
a disproportionate emphasis on optimization targets. 
Current research and industrial practices remain largely centered 
on minimizing latency at the request level, 
frequently at the expense of system-wide throughput and GPU utilization efficiency. 
This narrow focus leads to suboptimal GPU occupancy, 
constrained service capacity, and poor scalability under highly concurrent conditions.

The identified inefficiencies arise from several interrelated factors:
\begin{itemize}
    \item First, the multi-stage architecture of \AIGC{} pipelines 
    introduces diverse computational and memory requirements that differ substantially across stages. 
    Conventional monolithic system architectures often lack the granularity needed for fine-grained 
    resource allocation tailored to these varying demands.
    \item Second, the dynamic and often unpredictable nature of request patterns leads to significant 
    load fluctuations, which undermines the effectiveness of static resource provisioning strategies.
    \item Lastly, inefficient memory management exacerbates these challenges, 
    as intermediate results and cached data consume excessive VRAM without corresponding gains 
    in operational efficiency.
\end{itemize}

Microservices-based architectures offer a promising solution 
to these challenges by enabling more elastic and efficient resource utilization. 
For instance, Ant Group’s deployment of the NVIDIA Triton Inference Server 
for managing multi-model pipelines achieved a 2.4× improvement in throughput, 
a 20\% reduction in latency, and a 50\% decrease in operational costs~\cite{antgroup}.

The principal benefits of a microservices architecture include:
\begin{itemize}
    \item Elastic Resource Allocation: 
    Microservices can be deployed on hardware tailored to the specific computational 
    requirements of each component. 
    For example, compute-intensive tasks such as text encoding can utilize high-frequency GPUs, 
    while memory-intensive operations like frame synthesis benefit from GPUs with larger VRAM capacities.
    \item Independent Scalability: Each service can be scaled independently according 
    to real-time workload fluctuations. 
    This fine-grained scalability allows precise resource allocation, 
    reduces the need for over-provisioning, and lowers operational costs compared to monolithic systems.
    \item Fault Isolation and Enhanced Stability: The decoupled nature of microservices localizes failures 
    within individual components, preventing cascading system failures and improving overall reliability.
    \item Technological Heterogeneity: Each microservice can employ specialized frameworks, 
    programming languages, or optimizations best suited to its task, promoting flexibility 
    and performance gains.
\end{itemize}

However, disaggregating pipelines into discrete stages often introduces significant latency 
overheads due to the substantial volume of data transferred between nodes. 
Some existing systems that disaggregate transformer models,such as Mooncake~\cite{qin2024mooncake},
leverage NCCL~\cite{nccl} for message passing to mitigate this latency. 
Nevertheless, NCCL is primarily suitable for tensor-based messages 
with fixed payload sizes and cannot efficiently accommodate data with dynamic or variable dimensions.


In this paper, we present \sysname{}, 
a novel, large-scale distributed inference system designed 
for microservices-oriented deployment using RDMA and tailored to complex \AIGC{} workflows. 
By strategically decomposing end-to-end pipelines into fine-grained microservices, 
\sysname{} enables more flexible resource allocation and significantly improves GPU utilization. 
Additionally, we propose a systematic methodology for analyzing and identifying optimal 
pipeline partitioning strategies, 
thereby enhancing overall system performance and resource efficiency.

\sysname{} improves resource efficiency by sharing common pipeline components 
across multiple applications,
such as text-to-video (T2V) and image-to-video (I2V) generation will go though 
the same VAE Encoder and Decoder services.
To handle dynamic and unpredictable request loads, 
\sysname{} integrates a Node Manager (\NM{}) that dynamically redistributes GPU resources 
across microservices. 
This enables scalable capacity expansion during peak demand and contraction 
during low-traffic periods, ensuring high overall resource utilization.

To address the high data transfer latency associated with traditional TCP-based sockets 
in large-volume data scenarios, 
\sysname{} adopts one-sided RDMA for inter-service communication~\cite{jiang2004high, bedeir2010building, stuedi2012wimpy}. 
This approach supports direct memory-to-memory data transfer, 
substantially decreasing CPU overhead. 
However, due to the regional constraints inherent in RDMA connections, 
\sysname{} organizes services into regionally autonomous sets, 
each of which can execute complete workflows independently. 
Incoming requests are distributed randomly across these sets. 
Beyond improving resource utilization, this multi-set design also increases fault tolerance 
by isolating failures within specific regional sets.

As multiple instances deliver messages to a shared memory region within the same receiving instance via RDMA, 
\sysname{} employs a ring buffer structure that supports dynamically sized messages. 
However, deadlocks may occur if an instance fails during memory access. 
To the best of our knowledge, existing deadlock resolution mechanisms—such as those in Redis~\cite{carlson2013redis},rely on CPU involvement. 
To address this under RDMA constraints, 
\sysname{} introduces a novel double-ring buffer mechanism that resolves deadlocks without CPU intervention.

Our system delivers high throughput and low latency while maintaining high-quality content generation, 
providing a scalable, reliable, and cost-effective platform for real-world AIGC services.
In summary, this paper makes the following contributions:
\begin{itemize}
    \item We present the detailed design of \sysname{}, a large-scale distributed system 
    for orchestrating microservices-based \AIGC{} pipelines.
    \item We introduce a one-sided RDMA network architecture to enable high-performance 
    inter-service communication to reduce the latency from transferring high volumn data.
    \item We propose a deadlock-free multi-producer ring buffer data structure to facilitate 
    efficient message passing in the RDMA network with double ring buffers.
    \item We demonstrate a $16\times$ reduction in GPU resource usage for Wan2.1 
    image-to-video generation compared to running the pipeline within single instances.
\end{itemize}

\section{Background}

In this section, we introduce RDMA and NCCL, and then we distignish their difference.
\vspace{-2mm}
\subsection{RDMA (Remote Direct Memory Access)}
\label{ss:pre_rdma}
Traditional network communication typically involves multiple data copy operations: 
from application buffers to kernel space, 
and then to network interface buffers—all orchestrated by the CPU. 
This process introduces considerable latency and CPU overhead, particularly in distributed computing environments 
where fine-grained or frequent data exchanges are common.

Remote Direct Memory Access (RDMA)~\cite{jin2001high} enables a machine to directly access the memory of a remote machine without involving the remote CPU. 
By bypassing the kernel network stack, RDMA provides a low-latency, high-throughput communication mechanism. 
Communication via RDMA is facilitated through queue pairs, each consisting of a send queue and a receive queue. 
RDMA supports two primary operation types: two-sided and one-sided operations~\cite{kalia2016design}. 
Two-sided operations, such as Send and Receive, require participation from both the sender's and receiver's CPUs. 
In contrast, one-sided operations do not involve the remote CPU; 
the sender directly specifies the remote memory address for data placement. 
As a result, one-sided operations achieve lower latency and higher throughput compared to two-sided operations [23].


\subsection{NCCL (NVIDIA Collective Communications Library)}
\label{ss:pre_nccl}
NCCL (NVIDIA Collective Communications Library)~\cite{nccl} is a highly optimized communication library developed by NVIDIA, 
specifically designed to accelerate collective communication across multiple GPUs and nodes in deep learning workloads. 
It offers efficient primitives for data exchange and synchronization among GPUs, 
which are essential for distributed training and inference of large-scale models.

In distributed deep learning, GPUs collaborate by sharing intermediate data such as gradients and model parameters. 
To maintain high performance, it is critical to minimize communication overhead and latency through fast collective operations—especially 
as the number of GPUs scales. 
Conventional CPU-based communication methods often introduce bottlenecks due to limited PCIe bandwidth and excessive CPU involvement.

NCCL overcomes these limitations by implementing communication primitives entirely on the GPU 
and leveraging high-speed interconnects including NVLink, PCIe, and InfiniBand. 
It is engineered to maximize bandwidth efficiency and reduce latency, thereby supporting scalable distributed training.

As a foundational tool in modern AI infrastructure, NCCL significantly alleviates communication bottlenecks during gradient aggregation and synchronization. 
This enables near-linear scaling of training performance across numerous GPUs, reducing training time for increasingly large models and datasets.

\vspace{-2mm}
\subsection{Difference between RDMA and NCCL}
RDMA enables direct memory-to-memory data transfer across networked nodes without involving the CPU or operating system, 
thereby providing low-latency, high-bandwidth communication with minimal overhead. 
Its primary objective is to reduce latency and alleviate CPU and GPU resource consumption during remote data access.

In contrast, NCCL is optimized for high-performance GPU-to-GPU communication, both within and across nodes, 
leveraging hardware interconnects such as NVLink and InfiniBand. 
While NCCL can utilize RDMA as its underlying protocol for inter-node communication, 
it specifically offers tailored collective operations,such as All-Reduce and Broadcast,
that are designed for tensor data commonly used in deep learning workflows.

In \sysname{}, messages of arbitrary types must be transmitted between nodes without engaging GPU resources, 
as these are reserved for model computation tasks. 
Therefore, rather than relying on NCCL,which is oriented toward GPU centric tensor communication,
\sysname{} employs RDMA directly as its communication protocol. 
This approach ensures efficient, flexible, and non-intrusive data transfer while preserving valuable GPU capacity for core generative tasks.

\subsection{The AIGC Workflow}
\label{ss:wan}

In Artificial Intelligence Generated Content (\AIGC{}), a workflow denotes a structured, 
multi-stage pipeline that converts a user input,such as a text prompt or an image,into a refined output like video, image, or audio. 
Rather than a single step, this process typically involves a sequence of specialized AI models, 
each performing a distinct subtask. 
Common stages include input comprehension, iterative generation in a latent (compressed) space, 
and final decoding into a human-interpretable format.

To illustrate a typical \AIGC{} workflow, we consider the image-to-video generation process implemented by WAN~\cite{wan2.1}. 
This pipeline exemplifies the multi-stage architecture characteristic of advanced \AIGC{} systems, 
where each stage employs one or more dedicated models. 
The workflow comprises the following components:
\begin{itemize}
    \item T5 \& CLIP (Text Understanding and Conditioning): These models interpret and contextualize the text prompt, 
    converting it into a numerical representation that guides subsequent video generation.
    \item VAE Encode (Image Compression to Latent Space): This module compresses a high-dimensional input image into a compact latent representation, 
    optimizing computational efficiency.
    \item Diffusion (Iterative Video Generation in Latent Space): Serving as the core generative component, 
    this model utilizes the latent representation from the VAE Encoder and conditioned text embeddings to iteratively synthesize video frames.
    \item VAE Decode (Rendering the Final Video): This stage decodes the generated latent sequence back into a high-resolution pixel-domain video suitable 
    for viewing.
\end{itemize}

This modular workflow provides significant advantages in terms of efficiency and controllability, 
facilitating the production of sophisticated, high-fidelity video content from minimal input.

\section{System Design Overview}
\begin{figure}[t]
    \includegraphics[width=0.35\textwidth]{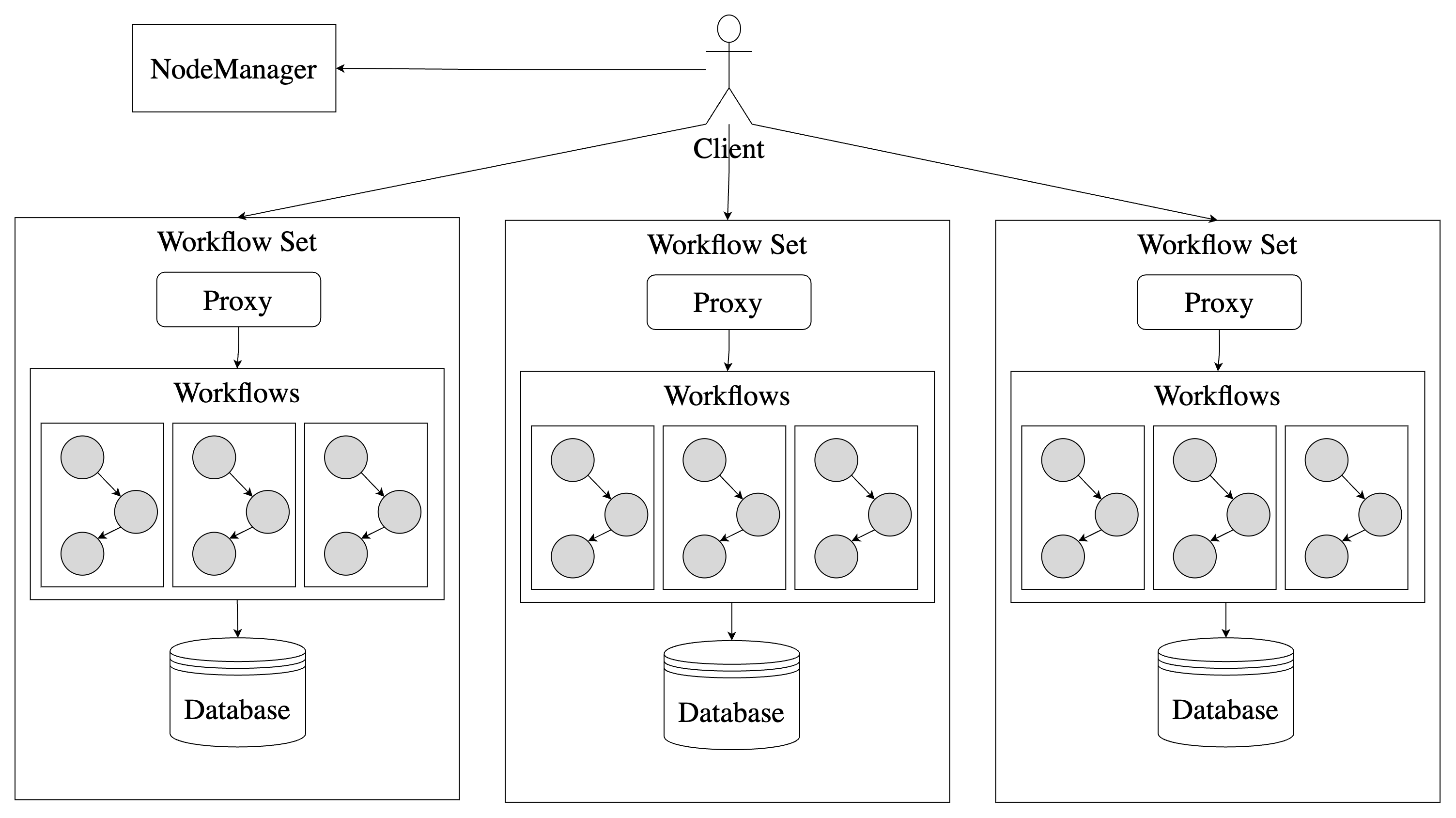}
    \caption{The system architecture overview of \sysname{}}
    \label{fig:overview}
    \vspace{-6mm}
\end{figure}

\sysname{} is a large-scale distributed inference system designed 
for \AIGC{} applications. 
It assists users in generating creative content to express aspects of their daily lives. 
The system organizes nodes (machines) into workflow sets, 
each deployable within a local RDMA-enabled network to leverage high-speed
and low-latency communication communication. 
Within each set, nodes are assigned one of three specialized roles: 
proxy, workflow, or database. 

Proxy nodes interface with client requests, workflow nodes execute the multi-stage 
\AIGC{} generation pipeline, 
and database nodes are responsible for persistent storage of generated results.

A centralized NodeManager (\NM{}) service maintains real-time information 
on the roles and network locations of all nodes. 
Clients first query the \NM{} to discover available proxy nodes, 
through which they submit generation requests. 
As the generation process can be time-intensive, 
clients periodically poll the system to retrieve results upon completion. 
Figure~\ref{fig:overview} illustrates the overall architecture of the system.

\subsection{Workflow Set (WS)}
\label{ss:ws}
\begin{figure*}[t]
    \includegraphics[width=0.7\textwidth]{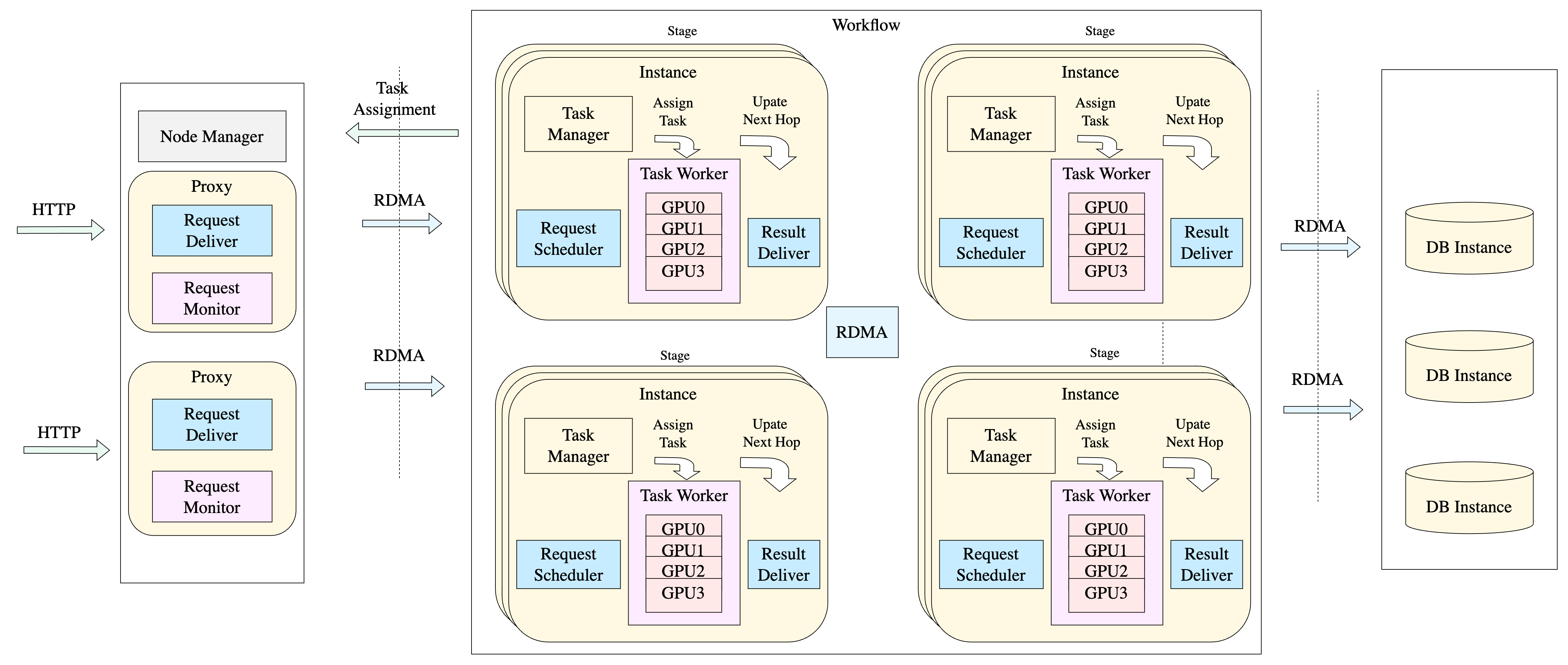}
    \caption{The Workflow Set in \sysname{}}
    \label{fig:workflow}
    \vspace{-4mm}
\end{figure*}

Each Workflow Set ($WS$) in \sysname{} comprises machines located within the same geographic region, 
interconnected via a dedicated RDMA network. 
As illustrated in Figure~\ref{fig:workflow}, a $WS$ consists of three core components: 
proxies, workflow instances, and databases. 
All communication between these components is facilitated through the high-performance RDMA network.

\subsection{Proxy}
\label{ss:proxy}

Within each $WS$, machines are interconnected via an RDMA network 
and require specialized hardware such as InfiniBand~\cite{liu2003high, islam2012high}. 
To facilitate external client access, proxy instances are deployed as lightweight, CPU-only servers. 
These proxies act as the entry point for generation requests: 
upon receiving a task, each proxy assigns a unique identifier (UID) 
that tracks the request throughout its entire lifecycle. 
This UID is propagated at every processing stage, enabling clients to retrieve results later 
by supplying the same identifier.

Each proxy also integrates a request monitor to regulate system load. 
It proactively rejects new incoming requests to avoid overloading. 
Clients that receive a rejection then attempt to submit their request to a different RDMA-enabled set. 
This monitoring mechanism, coupled with the fast-reject capability, 
enhances cross-set load balancing. 
Further design and operational details of the request monitoring system are elaborated 
in Section~\ref{ss:fast_reject}.


\subsection{Workflow Instances}
Upon acceptance by a proxy, 
each client request is processed through a user-defined workflow. 
Such a workflow consists of a sequence of stages,each defined by the user except for the first and last one,
where every stage contains multiple instances that handle requests in parallel.

The first stage, designated as the entrance stage, 
serves as the initial processing unit. 
It prepares and initiates the generation pipeline using the requests forwarded from the proxies. 
Once a request has been processed through all stages of the workflow, 
the final output is transferred to the database layer for persistent storage. 
This guarantees that all generated content is durably stored and remains retrievable by clients 
using the assigned UID.

Further details regarding workflow composition, 
instance varieties, 
scheduling strategies, 
and fault-tolerance mechanisms will be discussed in later sections.

\subsection{Databases}
A key design aspect of the \sysname{} database is its memory-centric storage strategy. 
Since most generated results are short-lived and typically accessed only once by clients, 
the system avoids writing data to disk by default. 
Instead, results are retained in distributed memory,such as RAM or NVMe-based buffers,
significantly accelerating both write and read performance. 
This approach reduces I/O overhead and minimizes latency during result submission and client retrieval.

To ensure fault tolerance and high availability, data is automatically replicated 
across multiple database instances within the same RDMA set. 
Given the transient nature of the generated materials, which become obsolete after client access, 
strong consistency consensus is not required for replication.

Each generated result,whether text, image, or video,is stored in the database alongside its associated 
UID for client retrieval. 
Once a client successfully fetches the result or after a predefined time-to-live (TTL) expires, 
the data is automatically purged. 
This efficient lifecycle management ensures that storage resources are prioritized 
for in-progress and recently completed tasks.

This lightweight and transient storage model enables \sysname{} to provide low-latency 
responses while upholding the reliability standards expected of a distributed inference system.

\section{Workflow Instances}
\label{ss:workflow}

In \sysname{}, an \AIGC{} workflow is a user-defined procedure designed to generate content 
through a sequence of AI models. 
As introduced in Section~\ref{ss:wan}, 
a typical workflow may involve generating a video from an image using a model such as Wan2.1~\cite{wan2.1}. 
This process begins with stages like T5 \& Clip and VAE Encode, 
which encode the input image into a latent tensor. 
A diffusion stage then iteratively samples this representation to produce an intermediate output, 
which is finally reconstructed into a video through a VAE Decode stage. 
Users can extend such workflows by incorporating additional models,
such as LoRA~\cite{sun2025one},to enhance output quality, 
introduce stylistic variations, or support more complex generation tasks.

A workflow instance in \sysname{} refers to a runtime entity 
that executes one or more models within the overarching \AIGC{} workflow, 
based on the user’s partitioning of the computational graph. 
For example, one instance may be dedicated to running the T5 \& Clip and VAE Encode stages, 
processing images and forwarding the encoded results. 
This design enables efficient resource utilization and pipeline parallelism.

Each workflow instance comprises four core components, 
as illustrated in Figure~\ref{fig:workflow}:
\begin{itemize}
    \item \textbf{TaskManager}: Initializes, coordinates, and monitors task execution within the instance.
    \item \textbf{RequestScheduler}: Manages incoming requests and assigns them to available workers.
    \item \textbf{TaskWorkers}: A pool of workers that perform actual model inference or data processing operations.
    \item \textbf{ResultDeliver}: Handles the forwarding of processed results to the next workflow stage or to the database for storage.
\end{itemize}

This component-based architecture promotes modularity, scalability, 
and a clear separation of concerns within each instance.

\begin{figure}[t]
    \includegraphics[width=0.4\textwidth]{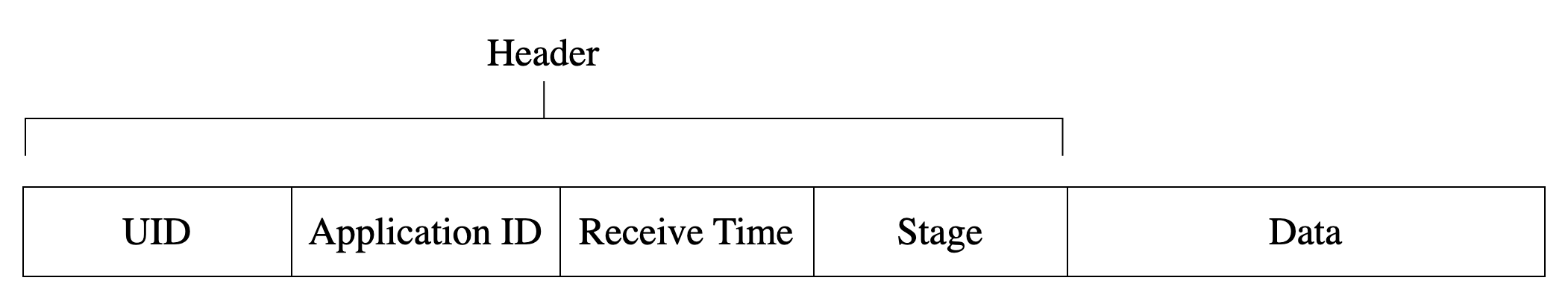}
    \caption{Workflow Message that contains a header and a value produced by the clients or the instances.}
    \label{fig:data}
    \vspace{-4mm}
\end{figure}

\subsection{Workflow Message}
\label{ss:data}
Before detailing each component, we first describe the structure of workflow messages. 
Each message consists of two parts: a payload and a header. 
The payload contains either the client’s original 
input or the output produced by a workflow instance, 
while the header carries metadata essential for request routing and processing. 
The overall message structure is depicted in Figure~\ref{fig:data}.

The header includes the following fields:
\begin{itemize}
    \item A UUID assigned by the proxy to uniquely identify the generation request throughout its entire lifecycle (as described in Section 3.2).
    \item A timestamp recorded by the proxy when the request is first received, used for monitoring task latency and system performance (Section 3.2).
    \item An application ID that specifies the processing logic and determines the next instance(s) to which the message should be routed (Section 4.5).
    \item The stage the messages is processing.
\end{itemize}

The overall structure is illustrated in Figure~\ref{fig:data}.

\subsection{TaskManager}
\label{ss:task_manager}

As each workflow instance is designed to execute arbitrary user-defined sub-models, 
a TaskManager component is responsible for dynamically determining which models the instance should run 
and identifying the subsequent instances to which results should be forwarded.

In \sysname{}, all workflow configurations are managed by the \NM{}. 
During initialization, the TaskManager queries the NM to retrieve the specific models 
and tasks assigned to its instance. 
Based on this assignment, it initializes the corresponding TaskWorkers to prepare for model inference. 
If no tasks are currently assigned,due to sufficient existing resource capacity,
the instance remains idle and may be repurposed for lower-priority operations, such as model training.

Concurrently, the TaskManager obtains routing information from the \NM{} that specifies the next hop in the workflow,
i.e., the downstream instances that should receive processed results. 
These routing details are configured into the ResultDeliver module (Section~\ref{ss:rd}), 
which handles efficient transmission of output data to subsequent instances over the RDMA network.

Additionally, the TaskManager maintains ongoing communication with the NM to stay synchronized 
with the latest task allocation policies, 
which may involve reassigning currently executing tasks to other instances. 
It periodically reports real-time GPU utilization metrics to the \NM{}, 
supporting dynamic load-aware scheduling. 
This feedback mechanism enables the \NM{} to redistribute tasks across available instances in a manner that optimizes 
overall GPU utilization and alleviates performance bottlenecks.

\begin{figure}[t]
    \center
    \begin{tabular}{c}
    \includegraphics[width=0.35\textwidth]{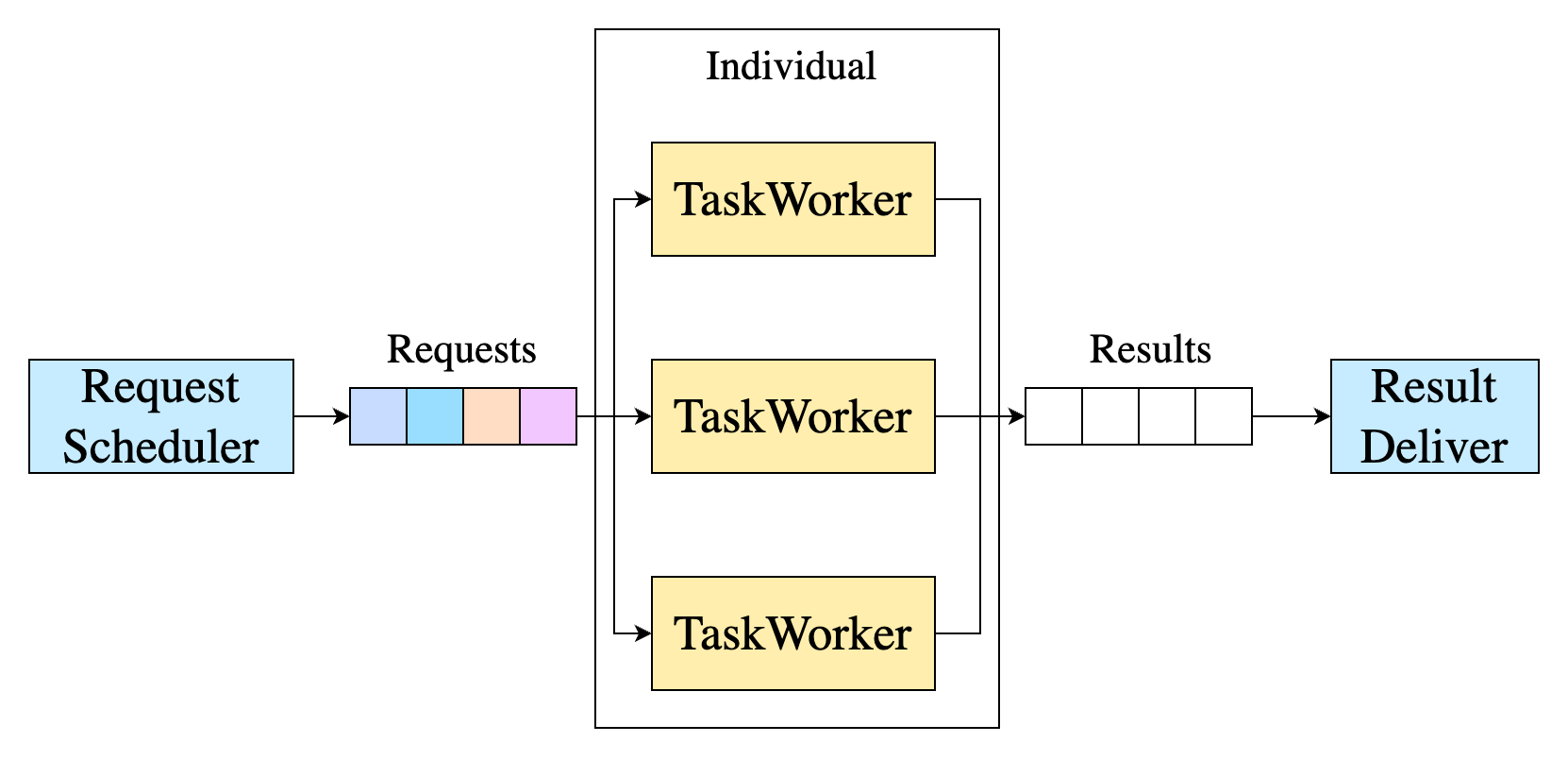} \\
        a) Individual Mode (IM) \\
    \includegraphics[width=0.35\textwidth]{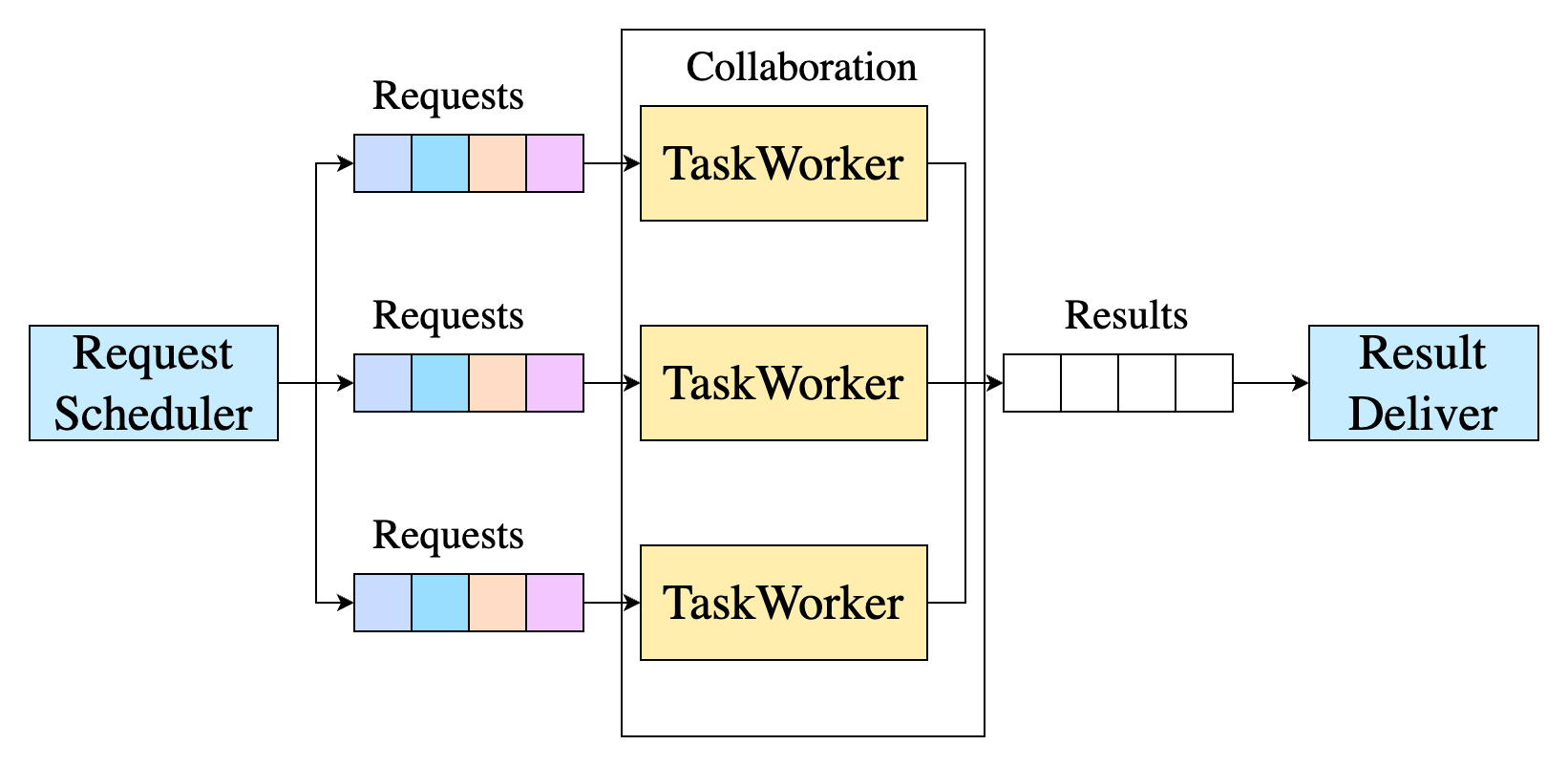} \\
        b) Collaboration Mode (CM)
     \end{tabular}
    \caption{Two modes of request assignments of RequestScheduler ($RS$).}
    \label{fig:rs_queue}
    \vspace{-5mm}
\end{figure}

\subsection{RequestScheduler ($RS$)}
\label{ss:rs}

To minimize latency in task processing, the RequestScheduler ($RS$) operates asynchronously, 
receiving incoming requests that are written directly into local memory via RDMA. 
This bypasses intermediate buffering and reduces processing overhead. 
The $RS$ continuously monitors a designated memory region for newly arrived requests. 
Once detected, each request is distributed to TaskWorkers according to the configured execution strategy.

\sysname{} supports two execution strategies:
\begin{itemize}
    \item Individual Mode ($IM$): In this mode, each worker processes requests independently using a dedicated GPU. 
        Instead of pushing requests directly to workers,which could cause load imbalance,
                the $RS$ maintains a shared local request queue (Figure~\ref{fig:rs_queue}a). 
                Idle workers autonomously fetch tasks from this queue, preventing any single worker from being overloaded while others remain idle. 
                This pull-based approach promotes natural load balancing.
    \item Collaboration Mode ($CM$): In this mode, all workers on a node cooperate to process a single request using all available GPUs. 
    When a new request arrives, the $RS$ broadcasts a copy to every worker (Figure~\ref{fig:rs_queue}b). 
    This ensures all workers receive the same inputs, making $CM$ suitable for distributed computation across multiple GPUs.
\end{itemize}

This flexible scheduling mechanism allows \sysname{} to efficiently handle diverse task granularities and computational requirements, 
enhancing system throughput and resource utilization.

\subsection{TaskWorker}

TaskWorker executes the requests delivered by the $RS$ according
to the predefined execution strategy. For certain tasks,such as data
preprocessing,that cannot utilize more than one GPU, each worker
processes requests individually to maximize parallelism. In contrast,
for computationally intensive operations like model sampling, the
system employs distributed computation strategies such as Pipeline
Parallelism ($PP$) or Tensor Parallelism ($TP$)~\cite{huang2019gpipe,wang2022tesseract,narayanan2019pipedream} to leverage
multiple GPUs, thereby reducing execution latency.
Within TaskWorker, the specific execution behavior is defined
by user-provided code. When a request is received, the TaskWorker
invokes the corresponding user function based on an application
identity attached to the request data, which specifies which appli-
cation logic should be executed.
To simplify model composition and interoperability, intermedi-
ate results can be represented in various data formats, including
tensors or raw binary data. Tensor data is placed directly into GPU
memory at the start of execution, enabling subsequent models to
access it immediately without unnecessary data movement. This
design minimizes transfer overhead and supports efficient pipeline
execution.

\subsection{ResultDeliver ($RD$)}
\label{ss:rd}
Upon completion of execution, TaskWorkers forward the results to the next instance in the workflow. 
In Individual Mode, each worker sends its results immediately after processing to minimize end-to-end latency. 
In Collaboration Mode, partial results from all workers are aggregated into a single consolidated output before delivery, 
which reduces communication overhead and ensures computational consistency.

The ResultDeliver ($RD$) component manages the routing of these outputs. 
$RD$ obtains routing information from the TaskManager (Section~\ref{ss:task_manager}), 
which in turn retrieves destination details from the \NM{}. 
Since a single instance may participate in multiple workflows, 
$RD$ uses the application identity included in the request (Section~\ref{ss:data}) 
to determine the appropriate next hop.

When multiple destination instances are available, $RD$ uses a round-robin mechanism to distribute results evenly. 
This prevents overloading any single downstream instance and improves load balancing across the system.

\section{Pipelining And Request Monitor}
\label{ss:fast_reject}

\begin{figure}[t]
    \includegraphics[width=0.48\textwidth]{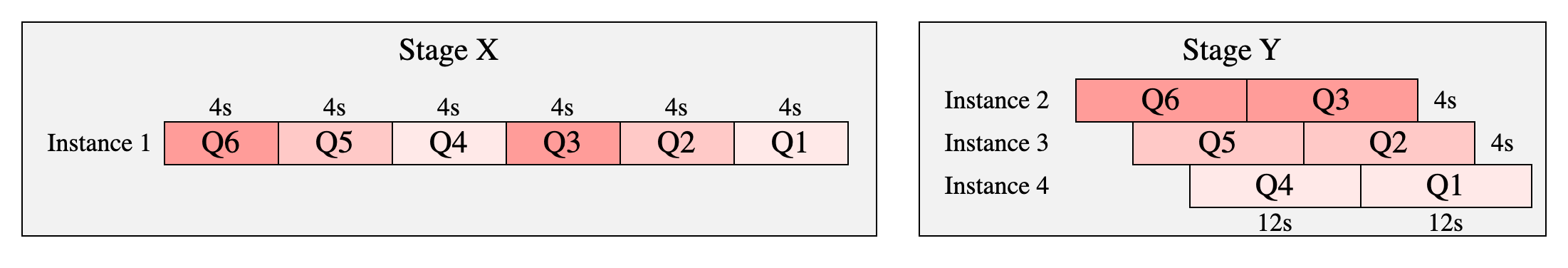}
    \caption{The pipeline example for one instance for stages $X$ and 3 instances for $Y$ with the execution cost $T_X=4s$ and $T_Y=12s$.
            Stage $X$ is applied with individual mode with 1 worker while stage $Y$ is applied with share mode.
            Stage $X$ accepts the requests and the stage $Y$ outputs the results every 4 seconds.
            }
    \label{fig:pipeline1}
    \vspace{-4mm}
\end{figure}

\begin{figure}[t]
    \includegraphics[width=0.48\textwidth]{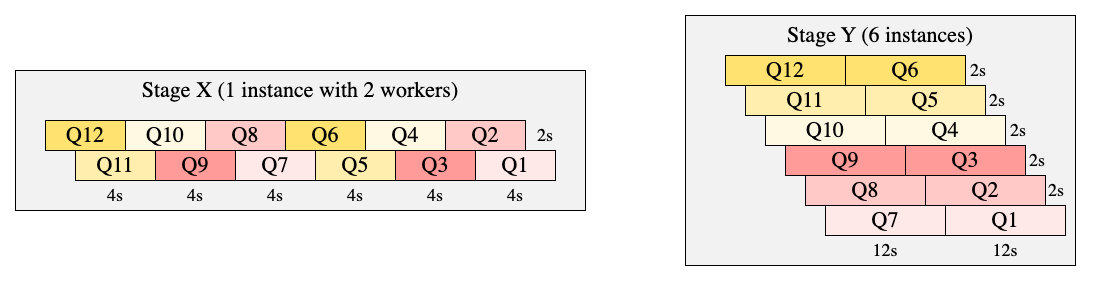}
    \caption{The pipeline example for one instance with 2 workers in stages $X$ and 6 instances in stage $Y$ with the execution cost $T_X=4s$ and $T_Y=12s$.
            Stage $X$ accepts the requests and the stage $Y$ outputs the results every 2 seconds.
            }
    \label{fig:pipeline5}
    \vspace{-4mm}
\end{figure}

Although separating different stages into distinct instances improves resource utilization, 
it can easily lead to high latency if the proxy continuously sends requests to the first stage, 
resulting in increased end-to-end delay for clients. 
To mitigate this, \sysname{} employs a pipelining technique.

Consider two stages, $X$ and $Y$, running on different instances, which take $T_X$ and $T_Y$ seconds to process a request, 
respectively. When $T_X \ge T_Y$, stage $X$ finishes processing before the next request arrives. We therefore focus on the case where $T_X < T_Y$.

To illustrate the pipelining mechanism, assume stage $X$ uses one instance with one worker in Individual Mode, 
while stage $Y$ is assigned $M = \lceil T_Y/T_X \rceil$ instances in Shared Mode. 
The proxy can then submit requests to $X$ every $_TX$ seconds, and stage $Y$ produces one result every $T_X$ seconds. 
This ensures that no request is delayed within the instances. 
The total latency for a request is $T(q) = T_X + T_Y + Network(q)$, where $Network(q)$ represents the inter-instance message transfer time.

For example, as shown in Figure~\ref{fig:pipeline1}, when $T_X = 4s$ and $T_Y = 12s$, 
stage $X$ accepts requests every $4$ seconds and outputs intermediate results at the same rate. 
Each instance in stage $Y$ receives inputs every $4$ seconds and produces final outputs after $12$ seconds, maintaining the input frequency.

We now extend this approach to multiple workers in stage $X$. 
If stage $X$ uses $K$ workers, we assign $M = \lceil K * T_Y / T_X \rceil$ instances to stage $Y$. 
In this case, the proxy can submit requests every $T_X/K$ seconds, 
and stage $Y$ will correspondingly produce results every $T_X/K$ seconds. 
Figure~\ref{fig:pipeline5} illustrates an example where stage $X$ uses $2$ workers and stage $Y$ uses $6$ instances.

\begin{theorem}
    \label{ss:latency}
    Consider two stages $X$ and $Y$ with execution times $T_X$ and $T_Y$ respectively, 
    where $T_X < T_Y$. If stage $X$ processes $K$ requests in parallel and stage $Y$ processes $M$ requests in parallel, where
    $M = \lceil T_Y/T_X * K\rceil$, then the output rates of stages $X$ and $Y$ will be equal.
\end{theorem}

\begin{proof}
    Since stage $X$ processes $K$ requests concurrently, 
    each taking time $T_X$, its throughput is $K/T_X$ outputs per second. 
    Equivalently, it produces one output every $T_X/K$ seconds, which also becomes the input rate for stage $Y$.

    Stage $Y$ can process $M$ requests in parallel. After receiving $M$ inputs, 
    the total time required for the first batch of $M$ outputs is at most $T_Y$ seconds. Since $M * T_X/K \ge T_Y$,
    the first output from $Y$ is produced no later than time $T_Y$. 
    Thereafter, for every $T_X/K$ seconds, one output is produced by $X$, and one output is completed by $Y$. 
    This is because one of the $M$ workers in $Y$ becomes available every $T_X/K$ seconds, matching the input rate from $X$.
    Hence, the output rate of $Y$ is $M/T_Y \ge K/T_X$,
    which ensures that the system reaches a steady state where the output rate of $Y$ equals the input rate from $X$. 
    Therefore, both stages produce outputs at the same rate of $K/TX$ per second.
\end{proof}

\textbf{Request Monitor and Fast Reject} 
To maintain stable processing latency, the proxy employs a fast-reject mechanism: 
if a client query arrives when the system is not ready to accept new requests, 
the proxy immediately rejects it to prevent excessive delays.

As established in Theorem~\ref{ss:latency}, 
it is always possible to assign a suitable value of $K$ to the next stage following the proxy such that the system produces outputs at a steady rate. 
The Request Monitor continuously calculates $K$ using real-time instance information obtained from the \NM{}.

Whenever the incoming request rate exceeds $K/T_X$, the proxy rejects additional requests. 
This fast-reject capability ensures that the system avoids overload and maintains stable performance under high load conditions.

\vspace{-2mm}
\section{RDMA Network}
The RDMA network (Section~\ref{ss:pre_rdma}) serves as a fundamental enabler in modern machine learning systems, 
largely due to its zero-copy capability that allows direct data transfer between memory regions without involving the host CPU. 
Systems such as Mooncake~\cite{qin2024mooncakekvcachecentricdisaggregatedarchitecture} and Deepseek~\cite{zhao2025insightsdeepseekv3scalingchallenges} 
often rely on NCCL (Section~\ref{ss:pre_nccl}) for message passing. 
However, despite its convenience, NCCL is unsuitable for \sysname{} due to several key limitations:

\begin{enumerate} [wide,nosep,label=(L\arabic*),ref={L\arabic*}]
    \item Limited Data Type Support: NCCL only supports tensor data, whereas messages in \sysname{} may contain arbitrary data types (Section~\ref{ss:data}).
    \item \label{c:nccl_c2} Fixed Message Size Requirement: NCCL requires fixed message sizes between endpoints, 
            while \sysname{} must accommodate dynamic message sizes resulting from variable user-defined workflows and models.
    \item GPU Interference: Because NCCL accesses GPU memory directly for transfers, it can interfere with concurrent computation and degrade worker performance.
    \item Lack of Message Context: NCCL uses end-to-end transmission semantics, which obscures message origin 
            and complicates the timely triggering of corresponding processing workflows.
\end{enumerate}

Due to these constraints, \sysname{} uses RDMA directly for message delivery. Nevertheless, employing bare-metal RDMA introduces two primary challenges:

\begin{enumerate} [wide,nosep,label=(C\arabic*),ref={C\arabic*}]
\item \label{c:rdma_c1} Memory Management for Dynamic Transfers: RDMA is a point-to-point protocol that requires establishing 
            a queue pair (QP) between communication peers. 
            The sender directly accesses the receiver’s memory via registered addresses. 
            Since registered memory regions are fixed in size, 
            effectively managing continuously arriving messages within finite memory becomes a critical issue.
\item \label{c:rdma_c2} Notification Mechanism for Data Arrival: \sysname{} uses one-sided RDMA operations, 
    where data arrival does not automatically notify the CPU. 
    Thus, a mechanism is required to alert the $RS$ when new messages arrive and indicate their location in memory.
\end{enumerate}

\begin{figure}[t]
    \includegraphics[width=0.45\textwidth]{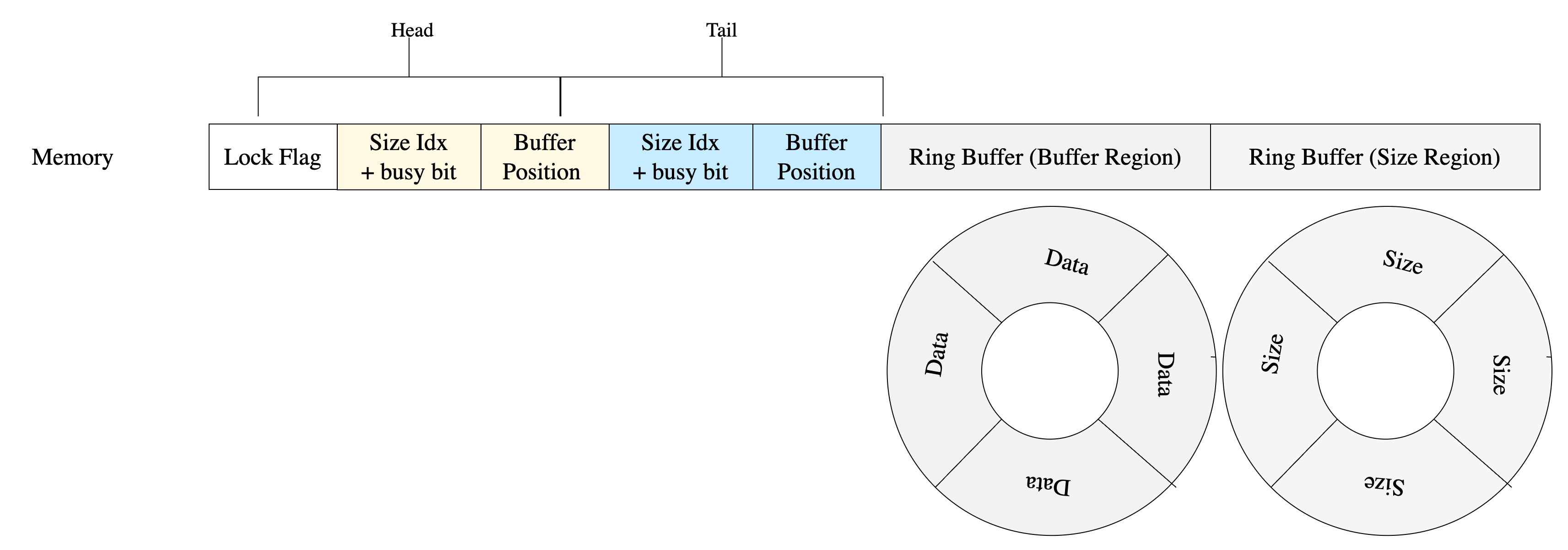}
    \caption{Ring Buffer}
    \label{fig:queue}
    \vspace{-4mm}
\end{figure}



\subsection{Ring Buffer}
To address challenge~\ref{c:rdma_c1}, 
\sysname{} employs a ring buffer (also referred to as a circular queue) inspired by~\cite{dragojevic2014farm}. 
For challenge~\ref{c:rdma_c2}, all senders share the same memory region, enabling the receiver to monitor only a single queue.

Locking is a commonly used mechanism to manage concurrent access to shared resources. 
However, in distributed systems, deadlock prevention becomes critical. 
Although Optimistic Concurrency Control (OCC)~\cite{kung1981optimistic} 
and two-phase locking (2PL)~\cite{soisalon1995partial} are widely adopted solutions that use versioning to protect data, 
they require CPU involvement and are thus unsuitable for the RDMA environment.

Although several non-blocking ring buffer implementations exist~\cite{harris2002practical, tsigas2001simple, kriz2013lock}, 
these are designed for fixed-size elements and do not accommodate variable message sizes, 
making them incompatible with \sysname{}'s requirements (\ref{c:nccl_c2}). 
To the best of our knowledge, no existing circular queue implementation adequately supports dynamic data assignment under our target constraints.

Our ring buffer supports multiple producers 
and a single consumer while accommodating dynamically sized messages. 
It is wait-free for the consumer; producers may contend but are only required to wait in cases of conflict. 
This ensures the consumer is never blocked: whenever new data is available in memory, 
it can be processed immediately. 
\textbf{We assume that consumer operations do not fail, as both the queue and the consumer are co-located.}

As illustrated in Figure~\ref{fig:queue}, the ring buffer’s memory structure consists of four components:
\begin{itemize}
\item \textbf{A lock region}, updated exclusively by senders using Compare-and-Swap (CAS) operations.
\item \textbf{A fixed-length header} containing the head and tail pointers that indicate where producers should write and consumers should read.
\item \textbf{A dynamically sized ring buffer region} that holds the actual message content.
\item \textbf{A size region} that records the size of each data entry in the buffer region. 
This region is essential for deadlock prevention (Section~\ref{ss:deadlock}). 
Each entry in the size region includes a busy bit indicating whether the slot has been written; 
this bit can only be cleared by the consumer.
\end{itemize}

We define the following operations for senders (producers) and the receiver (consumer).

\textbf{Sender (Producer) Operations}

To append data to the buffer, a sender performs the following steps:
\begin{enumerate}
\item Acquire the lock using a CAS-based spinlock.
\item Read the current tail position from the shared header. The header contains both the buffer tail (pointing to the next write location in the buffer region) and the size tail (pointing to the corresponding slot in the size region).
\item If insufficient space remains for the new data, release the lock and abort.
\item Check whether the next slot in the size region has been updated. 
If it has, update the header before writing new data. 
This situation may arise if a previous sender was lost after partial updates (see Proof Case~\ref{proof:case_7}).
\item Write the data into the buffer region starting at the buffer tail position.
\item Write the data size into the size region at the size tail position and set the busy bit.
\item Update the tail position in the header to reflect the new end of the data.
\item Release the lock.
\end{enumerate}

\textbf{Receiver (Consumer) Operations}

To consume data from the buffer, the receiver performs the following:
\begin{enumerate}
\item Read the current head position from the header.
\item If no new data is available, wait for a predefined interval and retry.
\item Read the next available data item from the buffer region.
\item Reset the busy bit in the size region.
\item Update the head position to effectively remove the consumed item.
\end{enumerate}

The ring buffer ensures that if a producer writes data starting at address $R$, 
the consumer will eventually read that same data from $R$, provided that no deadlock occurs.

\textbf{Buffer Region and Size Region}

Both the buffer region and the size region are implemented as ring buffers. 
The positions of read and write operations within these regions are managed by pointers stored in the header.

Due to the variable size of data elements, 
the pointer for the buffer region is updated dynamically based on the size value recorded in the corresponding slot of the size region:

\begin{equation*}
P_{b} = 
    \left\{
        \begin{aligned}
            &P_{b} + size(P_{size}), & if P_{b} + size(P_{size}) < RegionSize \\
            &0, & otherwise
        \end{aligned}
    \right.
\end{equation*}

The pointer for the size region, in contrast, is updated in fixed increments (by one unit per operation), wrapping around via modulo arithmetic:

\begin{equation*}
P_{size} = (P_{size} + 1)\mod \text{(RegionSize)}
\end{equation*}

Finally, both the pointers will be updated in the header in atomic operations.

\textbf{Deadlock and Liveness}
\label{ss:deadlock}

Unlike traditional ring buffers operating in multi-core environments, 
operations initiated by senders in a distributed setting may be lost, 
potentially leading to deadlocks. 
Most existing deadlock-mitigation strategies employ timer-based mechanisms~\cite{knapp1987deadlock, holt1972some, chandy1983distributed, singhal2002deadlock}. 
If a timer expires, the system assumes the sender has failed and releases the lock automatically.

However, ensuring liveness,i.e., that after a sequence of reads and writes the queue returns to a consistent state where the reader can retrieve valid data,
remains challenging. 
\sysname{} incorporates a dedicated \textit{size region} to help recover from deadlock scenarios.

Given the low latency of the RDMA network, \sysname{} uses a short timeout interval. 
When a timeout occurs, one of the senders releases and reacquires the lock. 
Nevertheless, delayed messages may still compromise reliability. 
Since \sysname{} is designed for time-sensitive workloads, such messages are allowed to overwrite current buffer entries. 
To detect data corruption, a checksum is applied to the data header. 
The consumer verifies this checksum upon reading; if a mismatch is detected, the data is discarded.

Thanks to the short timeout interval, obsolete updates can corrupt at most one subsequent data entry. 
Furthermore, since data is consumed shortly after being written, the time window during which messages remain vulnerable to corruption is very narrow. 
As a result, the probability of consecutive data corruption in \sysname{} is extremely low.

\textbf{Liveness}
We now enumerate possible execution scenarios and describe how the queue recovers liveness in each case.

To facilitate the discussion, we define the following atomic actions:
\begin{itemize}
\item $\text{Lock}(X)$: sender $X$ acquires the lock.
\item $\text{Unlock}(X)$: sender $X$ releases the lock.
\item $\text{WB}(X)$: sender $X$ writes data to the buffer at the tail.
\item $\text{WL}(X)$: sender $X$ writes the data size to the size region.
\item $\text{RB}(Z)$: receiver $Z$ reads data from the buffer.
\item $\text{RL}(Z)$: receiver $Z$ reads the data size.
\item $\text{GH}(X)$: $X$ reads the header and checks the size region.
\item $\text{UH}(X)$: $X$ updates the header.
\item $\text{TL}$: the lock times out.
\end{itemize}

\begin{enumerate} [wide,nosep,label=(Case\arabic*),ref={Case\arabic*}]
\item
$\text{Lock}(X) \rightarrow \text{TL} \rightarrow \text{Lock}(Y) \rightarrow \text{GH}(Y) \rightarrow \text{WB}(Y) \rightarrow \text{WL}(Y) \rightarrow \text{UH}(Y) \rightarrow \text{Unlock}(Y)$.
Here, $X$ is lost and $Y$ reacquires the lock. 
The receiver $Z$ will read valid data written by $Y$ and proceed to the next location.

\item 
$\text{Lock}(X) \rightarrow \text{GH}(X) \rightarrow \text{TL} \rightarrow \text{Lock}(Y) \rightarrow \text{GH}(Y) \rightarrow \text{WB}(Y) \rightarrow \text{WL}(Y) \rightarrow \text{UH}(Y) \rightarrow \text{Unlock}(Y) \rightarrow \text{WB}(X) \rightarrow \text{WL}(X)$.  
$X$ is delayed and overwrites $Y$'s data. $\text{WL}(X)$ fails due to the busy bit. If the data sizes from $X$ and $Y$ match, $Z$ reads valid data; otherwise, $Z$ may skip invalid entries and proceed using size metadata.

\item \label{proof:case_3}
$\text{Lock}(X) \rightarrow \text{GH}(X) \rightarrow \text{TL} \rightarrow \text{Lock}(Y) \rightarrow \text{GH}(Y) \rightarrow \text{WB}(Y) \rightarrow \text{WB}(X) \rightarrow \text{WL}(Y) \rightarrow \text{UH}(Y) \rightarrow \text{Unlock}(Y) \rightarrow \text{WL}(X)$.  
$X$ overwrites $Y$'s data late. $\text{WL}(X)$ fails since $\text{WL}(Y)$ occurred first. $Z$ reads using $Y$'s size and moves to the next location.

\item 
$\text{Lock}(X) \rightarrow \text{GH}(X) \rightarrow \text{TL} \rightarrow \text{Lock}(Y) \rightarrow \text{GH}(Y) \rightarrow \text{WB}(Y) \rightarrow \text{WB}(X) \rightarrow \text{WL}(X) \rightarrow \text{WL}(Y) \rightarrow \text{UH}(X) \rightarrow \text{Unlock}(X)$.  
$X$ updates the size before $Y$; $\text{WL}(Y)$ fails. $Z$ reads $X$'s data and continues.

\item 
$\text{Lock}(X) \rightarrow \text{GH}(X) \rightarrow \text{TL} \rightarrow \text{Lock}(Y) \rightarrow \text{GH}(Y) \rightarrow \text{WB}(X) \rightarrow \text{WB}(Y) \rightarrow \text{WL}(Y) \rightarrow \text{WL}(X) \rightarrow \text{UH}(Y) \rightarrow \text{Unlock}(Y)$.  
$X$ writes before $Y$, but $Y$ overwrites and finalizes the entry. $Z$ reads valid data from $Y$.

\item 
$\text{Lock}(X) \rightarrow \text{GH}(X) \rightarrow \text{TL} \rightarrow \text{Lock}(Y) \rightarrow \text{GH}(Y) \rightarrow \text{WB}(X) \rightarrow \text{WB}(Y) \rightarrow \text{WL}(X) \rightarrow \text{WL}(Y) \rightarrow \text{UH}(X) \rightarrow \text{Unlock}(X)$.  
Similar to Case~\ref{proof:case_3}: $X$ updates the size, but $Y$ overwrites the data. $Z$ skips invalid data and proceeds.

\item \label{proof:case_7}
$\text{Lock}(X) \rightarrow \text{GH}(X) \rightarrow \text{WB}(X) \rightarrow \text{WL}(X) \rightarrow \text{TL} \rightarrow \text{Lock}(Y) \rightarrow \text{GH}(Y) \rightarrow \text{UH}(Y) \rightarrow \text{WB}(Y) \rightarrow \text{WL}(Y) \rightarrow \text{Unlock}(Y)$.  
$X$ is lost after updating the size. $Y$ detects this, updates the header, and writes new data. $Z$ reads both $X$'s and $Y$'s data.

\item 
$\text{Lock}(X) \rightarrow \text{GH}(X) \rightarrow \text{WB}(X) \rightarrow \text{WL}(X) \rightarrow \text{UH}(X) \rightarrow \text{TL} \rightarrow \text{Lock}(Y) \rightarrow \text{GH}(Y) \rightarrow \text{WB}(Y) \rightarrow \text{WL}(Y) \rightarrow \text{UH}(X) \rightarrow \text{Unlock}(X)$.  
Normal case where $X$ does not release the lock. $Z$ still reads the data after $X$'s header update.

\end{enumerate}

In all scenarios, the receiver $Z$ eventually reaches a valid location and continues reading subsequent data correctly. Thus, the queue guarantees liveness.

\begin{theorem}
If a sender $X$ successfully writes data to a position $P$ in the buffer region, 
then the receiver $Z$ will eventually access $P$, but the data at $P$ is not guaranteed to be valid.
\end{theorem}

\begin{proof}
Assume, for the sake of contradiction, that $X$ writes data at position $P$ but $Z$ never reads from $P$.
For $Z$ to skip $P$, the size entry corresponding to $P$ in the size region must have been altered after $X$'s write. 
However, once $X$ successfully writes the data, it sets the busy bit in the associated size entry. 
This busy bit prevents any modification to that size entry until $Z$ consumes the data and clears the bit. 

Therefore, the size entry remains unchanged after $X$'s update, 
ensuring that $Z$ will eventually traverse the buffer region along the same logical path as $X$ and read the data from $P$.
\end{proof}
\section{Databases Replication}

\begin{figure}[t]
    \includegraphics[width=0.3\textwidth]{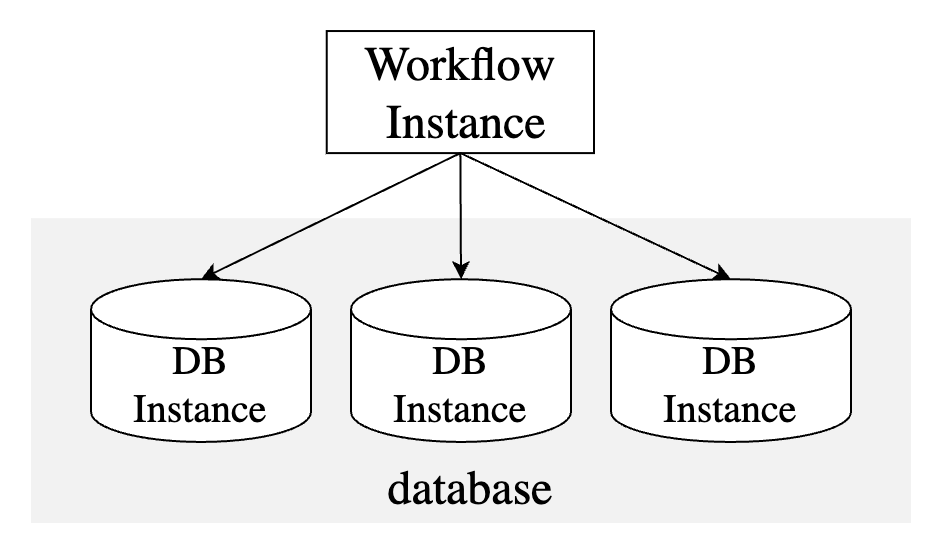}
    \caption{Database Replicatioin}
    \label{fig:db}
    \vspace{-4mm}
\end{figure}

Leveraging the reliable transmission capabilities of the RDMA network greatly simplifies the database replication process in \sysname{}. 
Moreover, due to the transient nature of \AIGC{}-generated results,
which are typically retrieved by clients within a short time window,
the system does not require outcomes to be stored persistently or guaranteed under strong consistency models.

Clients or proxies retrieve results by querying one database instance at a time. 
If the instance is operational and contains the requested result, 
the data is returned immediately. 
If the result is absent,for example, due to ongoing replication or instance failure,
the client proceeds to query another instance in the next attempt. 
This lightweight approach ensures robustness and high availability without incurring the overhead of distributed consensus or synchronous replication.

\section{NodeManager (\NM{})}

\begin{figure}[t]
    \includegraphics[width=0.4\textwidth]{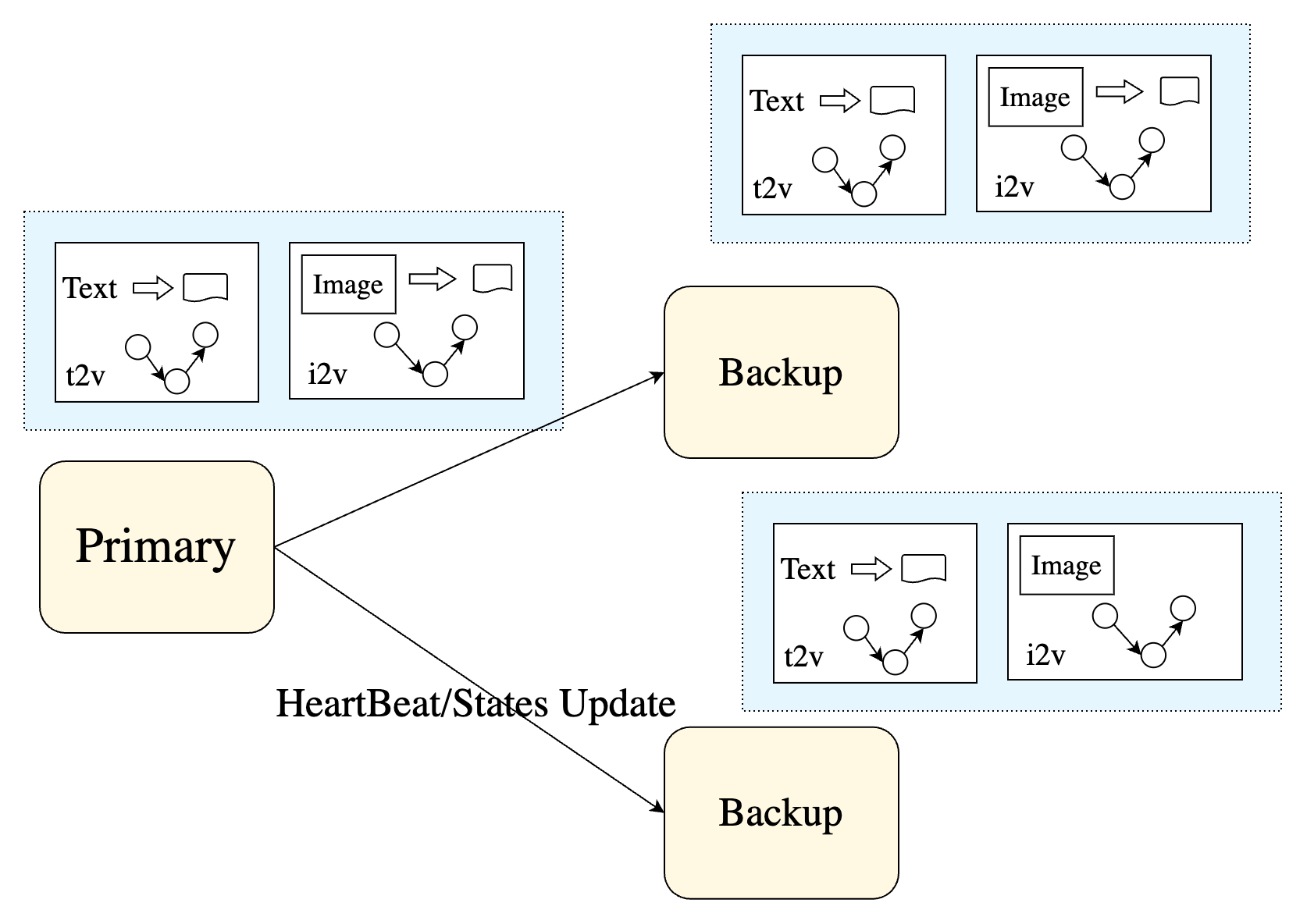}
    \caption{NodeManager}
    \label{fig:node_manager}
    \vspace{-4mm}
\end{figure}

The Node Manager (\NM{}) acts as a centralized orchestrator that maintains metadata about all instances and their roles within each workflow set. 
To communicate with any proxy or database instance, clients must first query the \NM{} to obtain the respective instance’s location.

Beyond role and location management, the \NM{} continuously monitors the GPU utilization of each workflow stage. 
If a stage exhibits insufficient GPU capacity, the \NM{} automatically scales the stage by assigning additional instances, 
thereby enhancing computational throughput and ensuring timely task completion.

\subsection{Primary Election}

The \NM{} employs a primary-backup replication scheme to ensure high availability. 
All read queries are directed to the primary instance, while write operations are propagated to the backup instances, 
as illustrated in Figure~\ref{fig:node_manager}.

As in conventional distributed systems, the primary \NM{} instance periodically broadcasts heartbeats to maintain its presence and authority. 
If any instance detects the absence of heartbeats from the current leader,or the lack of a leader altogether,
it initiates a new leader election using the Paxos consensus algorithm~\cite{paxos}.

The Paxos protocol guarantees that at most one leader is elected at any given time, 
even under concurrent election attempts, 
thereby ensuring system safety. 
This Paxos-based election mechanism provides fault tolerance and strong consistency. 
During elections, instances exchange status information to ensure the new leader incorporates the latest system state, 
preserving metadata consistency across the management layer.

\begin{figure}[t]
    \includegraphics[width=0.4\textwidth]{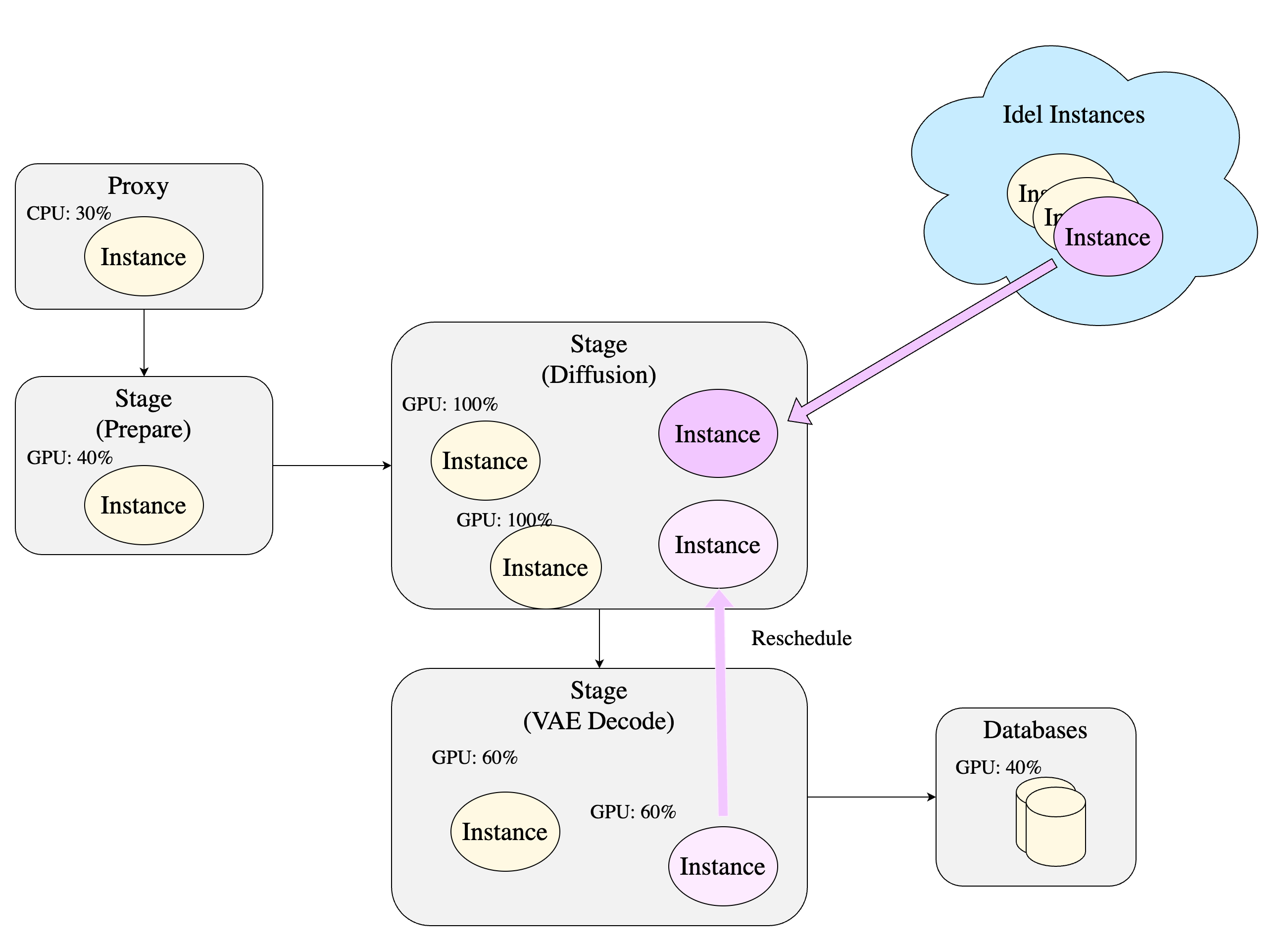}
    \caption{Reschedule instances from VAE Decode and Idel Instance pool to Diffusion.}
    \label{fig:reschedule}
    \vspace{-5mm}
\end{figure}

\subsection{Workflow Instance Assignment}

\sysname{} concurrently executes numerous application workflows, 
each of which may require varying amounts of GPU resources at different stages due to fluctuating and imbalanced request loads. 
This variability necessitates dynamic reassignment of GPU resources to improve overall utilization.

The resource assignment process consists of the following steps:
\begin{itemize}
\item \textbf{Instance GPU Status Reporting:} Each instance periodically reports its GPU utilization to the \NM{}.
\item \textbf{Average GPU Utilization Calculation:} The \NM{} computes the average GPU utilization for each workflow stage within the same Workflow Set (Section~\ref{ss:ws}) over a recent time window (e.g., 5 minutes).
\item \textbf{Busy Stage Identification:} The \NM{} identifies the stage with the highest utilization.
\item \textbf{Instance Assignment:} If the utilization exceeds a predefined threshold (e.g., 85\%), the \NM{} assigns additional instances to that stage.
\item \textbf{State Delivery:} The \NM{} updates the instance’s role and sends the corresponding task and routing information (including next-hop details) 
                to the assigned instance.
\item \textbf{Local State Update:} If the task assignment changes, the instance initializes the corresponding models and updates the next-hop information 
            within its $RD$ component.
\end{itemize}

In addition to active instances processing workflow stages, \sysname{} maintains an Idle Instance Pool. 
These instances may be used for lower-priority tasks, such as model training. 
When GPU resources become scarce, idle instances can be reassigned to support overloaded workflow stages.

Figure~\ref{fig:reschedule} illustrates an example of dynamic instance rescheduling during video generation. 
A request passes through a preparation stage, a diffusion model, and finally a VAE decoder that produces the video. 
In this scenario, an underutilized instance (60\% GPU) in the preparation stage is reassigned to the diffusion stage, 
where all instances are at full capacity. This decision is triggered by regular GPU status reporting.

This instance assignment mechanism enables efficient workload balancing both within and across workflow sets.

\begin{figure}[t]
    \includegraphics[width=0.4\textwidth]{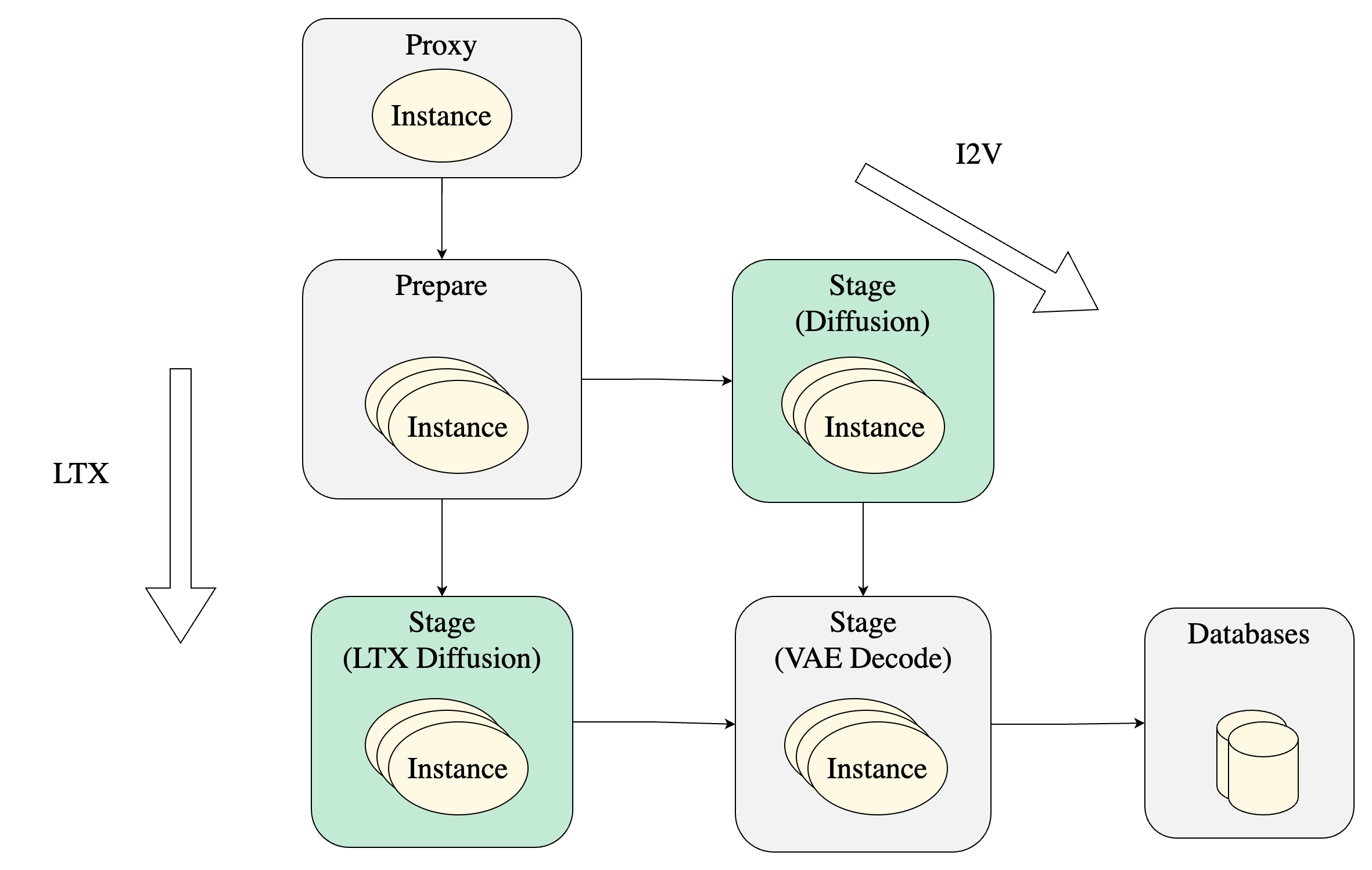}
    \caption{Two workflows share some overlap stages.}
    \label{fig:share}
    \vspace{-5mm}
\end{figure}

\subsection{Instance Sharing}
To improve resource utilization, \sysname{} allows multiple workflows to share computational instances. 
For example, the LTX model,which generates videos from multiple images,
and the image-to-video workflow,which produces videos from a single image,
share all processing stages except their respective diffusion models, 
as illustrated in Figure~\ref{fig:share}. 
This sharing of common stages reduces redundant resource allocation while preserving workflow-specific processing where necessary.
%
%
%

\section{Fault Tolerance}
\sysname{} provides fault tolerance primarily for the \NM{}, 
reflecting the transient nature of \AIGC{} generation tasks. 
If messages are lost between workflow stages,whether due to network issues or data errors (Section~\ref{ss:deadlock}),
the system does not attempt retransmission. 
This design stems from the observation that retries would introduce additional latency, 
and clients or end-users in interactive \AIGC{} applications are generally unwilling to tolerate extended delays. 
By omitting retry mechanisms, 
\sysname{} achieves a simpler and higher-performance architecture. 
The messaging layer can, however, be extended to support guaranteed delivery for applications that require it.

\vspace{-2mm}
\section{Related Work}


\textbf{Inference Acceleration}
Training-based acceleration methods primarily aim to reduce sampling steps~\cite{salimans2022progressive, yue2023resshift, luo2023latent, sauer2024adversarial, yin2024one} 
or optimize model architectures~\cite{li2023snapfusion, yang2023diffusion}. 

However, these techniques focus predominantly on training-related costs and algorithmic complexity. 
In contrast, real-world AIGC applications typically involve multi-model workflows to generate content. 
\sysname{} addresses system-level efficiency by dynamically allocating more resources to critical sub-models, 
thereby accelerating end-to-end pipeline performance.

\textbf{Inference on Disaggregated Architecture}

Significant efforts have been dedicated to improving the efficiency of LLM serving systems through scheduling, memory management, 
and multi-stage resource optimization. 
Systems such as Orca~\cite{yu2022orca} employ iteration-level scheduling to enable concurrent processing across different stages, 
while FlexGen~\cite{sheng2023flexgen}, SARATHI~\cite{agrawal2024taming}, 
and FastServe~\cite{wu2023fast} introduce innovative scheduling 
and swapping strategies to better distribute workloads across limited hardware.

Approaches like Mooncake~\cite{qin2024mooncake} and Splitwise~\cite{patel2024splitwise} further improve throughput by disaggregating 
the prefill and decoding stages. 

Inspired by these works, \sysname{} applies a disaggregated architecture to \AIGC{} workflows by partitioning models across different 
hardware resources, with the goal of maximizing overall GPU goodput.

\textbf{Distributed Systems with RDMA}

Numerous research efforts~\cite{wang2022case, kalia2014using, ma2022survey, poke2015dare, cai2018efficient, luo2024splitft, taleb2018tailwind, zhang2022ford} 
have adopted RDMA as the underlying communication protocol in distributed systems. 
These systems typically rely on primary-based protocols to address distributed locking challenges. 
In contrast, \sysname{} does not depend on any primary node for message delivery and instead employs a distributed locking mechanism to manage coordination.

\textbf{Multiple-Producer Ring Buffer}

Multiple-producer ring buffers have been extensively studied~\cite{feldman2015wait, lee2009lock, barrington2015scalable}. 
Although many existing solutions claim to be wait-free, they assume fixed-size buffer slots. 
\sysname{} relaxes this constraint by employing a lock-based mechanism for producers—augmented with deadlock detection—while still 
ensuring a wait-free experience for the consumer.

Systems such as Chubby~\cite{burrows2006chubby} and Redis~\cite{carlson2013redis} offer solutions for deadlock resolution in distributed locking, 
but these typically involve CPU-based coordination. 
\sysname{} introduces a novel approach that resolves deadlocks efficiently within the constraints of RDMA protocols.

\vspace{-2mm}
\subsection{Conclusion}

\sysname{} is a large-scale distributed inference system designed for AIGC workflows. 
It disaggregates workflow stages across multiple instances to improve GPU utilization and employs RDMA technology 
to reduce communication latency caused by large payload transfers between nodes. 
Additionally, \sysname{} introduces a novel solution to resolve deadlocks in shared ring buffers. 
Experimental results demonstrate that the system achieves a 16\% reduction in GPU resource consumption for Wan2.1 image-to-video generation tasks.

\bibliographystyle{acmart-primary/ACM-Reference-Format}
\bibliography{references}

\end{document}